\documentclass[11pt,a4paper]{article}
\usepackage[utf8]{inputenc}

\usepackage{amsmath}
\usepackage{amsfonts}
\usepackage{amssymb}

\usepackage{amsthm}
\newtheorem{thm}{Theorem}
\newtheorem{cor}{Corollary}
\newtheorem{lem}{Lemma}
\newtheorem{prop}{Proposition}

\usepackage{multicol}
\usepackage{listings}
\author{Janosch Döcker}

\title{\textsc{Monotone 3-Sat-$(2,2)$} is NP-complete}

\begin{document}

\maketitle

\begin{abstract}
We show that \textsc{Monotone 3-Sat} remains NP-complete if (i) each clause contains exactly three distinct variables, (ii) each clause is unique, i.e., there are no duplicates of the same clause, and (iii), amongst the clauses, each variable appears unnegated exactly twice and negated exactly twice. Darmann and Döcker~\cite{darmann19} recently showed that this variant of \textsc{Monotone 3-Sat} is either trivial or NP-complete. In the first part of the paper, we construct an unsatisfiable instance which answers one of their open questions (Challenge 1) and places the problem in the latter category. 

Then, we adapt gadgets used in the construction to (1) sketch two reductions that establish NP-completeness in a more direct way, and (2), to show that $\forall\exists$ \textsc{3-SAT} remains $\Pi_2^P$-complete for quantified Boolean formulas with the following properties: (a) each clause is monotone (i.e., no clause contains an unnegated and a negated variable) and contains exactly three distinct variables, (b) each universal variable appears exactly once unnegated and exactly once negated, (c) each existential variable appears exactly twice unnegated and exactly twice negated, and (d) the number of universal and existential variables is equal. Furthermore, we show that the variant where (b) is replaced with (b') each universal variable appears exactly twice unnegated and exactly twice negated, and where (a), (c) and (d) are unchanged, is $\Pi_2^P$-complete as well. Thereby, we improve upon two recent results by Döcker et al.~\cite{doecker19} that establish $\Pi_2^P$-completeness of these variants in the non-monotone setting.   

We also discuss a special case of \textsc{Monotone 3-Sat-$(2,2)$} that corresponds to a variant of \textsc{Not-All-Equal Sat}, and we show that all such instances are satisfiable. 
\end{abstract}


\section{Introduction}

The satisfiability problem for Boolean formulas is one of the go-to problems when choosing a base problem for polynomial reductions. Indeed, it was the first problem shown to be NP-complete~\cite{cook71}. The seminal book by Garey and Johnson~\cite{garey79} contains a large list of known NP-complete problems and an extensive introduction into the theoretical foundation of NP-completeness. A very popular variant of the satisfiability problem is \textsc{3-SAT}, where each clause contains exactly three variables. This problem remains NP-complete even if further restrictions are imposed (see Table~\ref{table:complexity_results_overview}). In this article, we consider variants of \textsc{3-SAT} where each clause contains exactly three \emph{distinct} variables. Hence, unless we explicitly say otherwise, the considered instances have this property (the same goes for references regarding \textsc{3-SAT} variants).   

\begin{table}[h]
\centering
\begin{tabular}{|c|c|c|c|c|c|c|c|}
\hline 
\multicolumn{2}{|c|}{Clauses} & \multicolumn{4}{|c|}{Variables} & Complexity\\
\cline{1-7}
unique & monotone & E4 & 3P1N, 1P3N & 3P1N & 2P2N & \\ 
\hline 
\hline 
\checkmark &  &  &  & \checkmark &  & NP-c~\cite[Cor.\,11]{darmann19} \\ 
\hline 
\checkmark &  &  &  &  & \checkmark & NP-c~\cite[Thm.\,1]{berman03}\\ 
\hline 
\checkmark & \checkmark & \checkmark &  &  &  & NP-c~\cite[Cor.\,4]{darmann18} \\
\hline 
\checkmark & \checkmark &  & \checkmark &  &  & NP-c~\cite[Thm.\,9]{darmann19} \\
\hline 
\checkmark & \checkmark &  &  & \checkmark &  & ? \\
\hline 
 & \checkmark &  &  &  & \checkmark & NP-c~\cite[Thm.\,5]{darmann19} \\
\hline 
\checkmark & \checkmark &  &  &  & \checkmark & NP-c (Thm.~\ref{thm:mon_3sat-(2,2)}) \\
\hline 
\end{tabular} 
\caption{Overview of complexity results for (monotone) \textsc{3-SAT}. A checkmark in the ``unique'' subcolumn means that each clause contains exactly three \emph{distinct} variables. The headings of the subcolumns in the ``Variables'' column denote the following properties: E4 := each variable appears exactly four times; 3P1N, 1P3N := each variable appears exactly four times and either exactly once unnegated or exactly once negated; 3P1N := each variable appears exactly three times unnegated and once negated; 2P2N := each variable appears exactly twice unnegated and exactly twice negated. In the last column we use the abbreviation NP-c for NP-complete. Note that we only ticked the strongest restrictions, e.g., a checkmark in the 3P1N subcolumn implies a checkmark in the two preceding subcolumns. Moreover, by symmetry we can omit the 1P3N case (identical to 3P1N).}
\label{table:complexity_results_overview}
\end{table}

Recently, Darmann and Döcker~\cite[Cor.\,2]{darmann19} showed that for each fixed $k \geq 3$ \textsc{Monotone 3-Sat} is NP-complete if each variable appears exactly $k$ times unnegated und exactly $k$ times negated. Further, they were able to prove that the case $k = 2$ is either trivial or NP-complete. In other words, finding a single unsatisfiable instance is enough to prove that the problem remains NP-complete for $k = 2$. Hence, by constructing an unsatisfiable instance for $k = 2$, we settle this case and thus, one of their open problems (Challenge~1). As the problem is trivial for $k = 1$~\cite[p.\,32]{darmann19} by a result from Tovey~\cite[Thm.\,2.4]{tovey84}, our result closes the last remaining gap for this variant of \textsc{Monotone 3-Sat}. 

The gadgets used in the construction of the unsatisfiable instance can also be used to obtain a more direct way of establishing NP-completeness for the case $k = 2$ (we describe two reductions in this article). Then, we use one of the new gadgets to show that two recent results from Döcker et al.~\cite[Thm.\,3.1 and Thm.\,3.2]{doecker19} hold even in the monotone setting. First, we show that $\forall\exists$ \textsc{3-SAT} remains $\Pi_2^P$-complete if (i) each clause is monotone (ii) each universal variable appears exactly once unnegated and exactly once negated, (iii) each existential variable appears exactly twice unnegated and exactly twice negated, and (iv) the number of universal and existential variables is equal. Second, we show that the variant where (ii) is replaced with (ii') each universal variable appears exactly twice unnegated and exactly twice negated, and where (i), (ii) and (iv) are unchanged, is $\Pi_2^P$-complete, too. 

The article is structured as follows: In Section~\ref{sec:preliminaries}, we recall important definitions and concepts. Then, in Section~\ref{sec:construction_unsatisfiable_instance}, we construct an unsatisfiable instance of \textsc{Monotone 3-Sat-$(2,2)$}. Section~\ref{sec:more_ways} contains two reductions that can be used to obtain the main result in a more direct way and one of the involved gadgets is subsequently used in Section~\ref{sec:quantified_results} to show that a restricted variant of $\forall\exists$ \textsc{3-SAT} remains $\Pi_2^P$-complete. The appendix contains proofs of two Lemmas used in Section~\ref{sec:construction_unsatisfiable_instance}, and a representation of 
\begin{itemize}
\item a gadget on which several of our results are based, and 
\item the constructed unsatisfiable instance of \textsc{Monotone 3-Sat-$(2,2)$},
\end{itemize}
which can be used to verify our results with the help of a SAT Solver (e.g., using the PySAT Toolkit~\cite{ignatiev18}). 

\section{Preliminaries}\label{sec:preliminaries}

Let $V=\{x_1,x_2,\ldots,x_n\}$ be a set of $n$ variables. We also write $X_1^i$ to denote the set $\{x_1,x_2,\ldots,x_i\}$ for $i \geq 1$. A positive literal is an element of $\mathcal{L}_+ = V$, a negative literal is an element of $\mathcal{L}_- = \{\overline{x_i} \mid x_i \in V\}$, and the set of literals is denoted by $\mathcal{L} = \mathcal{L}_+ \cup \mathcal{L}_-$. A clause is a subset of $\mathcal{L}$. We say that a clause $C_j \subseteq \mathcal{L}$ is a $k$-clause if $|C_j| = k$ and $C_j$ is monotone if $C_j \subseteq \mathcal{L}_+$ or $C_j \subseteq \mathcal{L}_-$. A Boolean formula is a set of $m$ clauses
\[
\bigcup_{j=1}^m \{C_j\}. 
\]
A Boolean formula is monotone if $C_j$ is monotone for each $j \in \{1, \ldots, m\}$. A truth assignment $\beta\colon V \rightarrow \{T, F\}$ maps each variable to the truth value $T$ (True) or $F$ (False). A formula is satisfied for a truth assignment $\beta\colon V \rightarrow \{T, F\}$ if $\beta$ sets at least one literal in each clause true (e.g., a negative literal evaluates to true if $\beta$ sets the corresponding variable false). If such a truth assignment exists, we say that the formula is satisfiable; otherwise the formula is unsatisfiable. Further, a formula is nae-satisfiable if and only if there exists a truth assignment $\beta$ that sets at least one literal in each clause true and at least one false. The main result concerns the following decision problem. 

\begin{center}
\noindent\fbox{\parbox{.95\textwidth}{
\noindent \textsc{Monotone 3-Sat-$(2,2)$}\\
\noindent{\bf Input.} A Boolean formula
$$\bigcup_{j=1}^m \{C_j\}$$
over a set  $V=\{x_1,x_2,\ldots,x_n\}$ of variables such that (i) each $C_j$ is a unique monotone 3-clause that contains exactly three \emph{distinct} variables, and (ii), amongst the clauses, each variable appears unnegated exactly twice and negated exactly twice.
\\
\noindent{\bf Question.} Does there exist a truth assignment for $V$ such that each clause of the formula is satisfied?
}}
\end{center}

\noindent {\bf Remark.} A monotone 3-clause always contains exactly three distinct variables. 

In one instance, we reduce from \textsc{Monotone 3-Sat*-$(2,2)$}~\cite{darmann19} which is the variant of \textsc{Monotone 3-Sat-$(2,2)$} where variables may appear more than once in a clause. Note that we can assume that each variable appears at most twice in a given clause, since each clause is monotone and there are only two unnegated and two negated appearances of any variable. 

\noindent {\bf Enforcers.} In the construction of an unsatisfiable instance of \textsc{Monotone 3-Sat-$(2,2)$} and the reductions after that, we make use of gadgets that enforce truth assignments to have certain properties (gadgets also go by the name of \emph{enforcers}~\cite{berman03}). As an example, we consider an enforcer introduced by Berman et al.~\cite[p.\,3]{berman03}:
\begin{align*}
\mathcal{S}(\ell_1, \ell_2, \ell_3) = &(\ell_1 \vee \overline{a} \vee b) \wedge (\ell_2 \vee \overline{b} \vee c) \wedge (\ell_3 \vee a \vee \overline{c}) \wedge {} \\ 
&(a \vee b \vee c) \wedge (\overline{a} \vee \overline{b} \vee \overline{c}), 
\end{align*}
where $a, b, c$ are new variables. The enforcer $\mathcal{S}(\ell_1, \ell_2, \ell_3)$ can not be satisfied by a truth assignment $\beta$ that sets all literals in $\{\ell_1, \ell_2, \ell_3\}$ false. On the other hand, if at least one literal in $\{\ell_1, \ell_2, \ell_3\}$ evaluates to true, we can find truth values for the variables $a, b, c$ such that all clauses of the enforcer are satisfied. In other words, $\mathcal{S}(\ell_1, \ell_2, \ell_3)$ simulates a clause but has the advantage that we can allow duplicates since each literal in $\{\ell_1, \ell_2, \ell_3\}$ ends up in a different clause (cf.~\cite[p.\,3]{berman03}). Note that this enforcer is not monotone. In this article, we construct a monotone version with 99 new variables and 133 clauses (instead of 3 new variables and 5 clauses in the setting above). 

\section{Construction of an unsatisfiable instance of \textsc{Monotone 3-Sat-$(2,2)$}}\label{sec:construction_unsatisfiable_instance}

In this section, we construct an unsatisfiable instance of \textsc{Monotone 3-Sat-$(2,2)$}. First, we construct an enforcer~$\mathcal{M}^{(i)}(u_1, \overline{u_2}, \overline{u_3})$ that, intuitively, consists of three smaller gadgets. The first gadget is only satisfiable by truth assignments for the corresponding variables that can be placed in one of two categories. Depending on the category of the truth assignment (and the restrictions imposed by them), it is not possible to find a truth assignment for the variables contained in the second or the third gadget such that all clauses are satisfied. The second and the third gadget (see Lemmas~\ref{lem:second_gadget} and~\ref{lem:third_gadget}) have been found via computer search. The basic idea of the implemented Python code is the following: start with a collection of random candidates and try to improve them by swapping literals of differenct clauses, where this operation preserves the properties of an instance of \textsc{Monotone 3-Sat-$(2,2)$} (a reduction in the number of satisfying truth assignments is considered an improvement here). We used the PySAT Toolkit~\cite{ignatiev18} to (1) obtain a list of all satisfying truth assignments for a given collection of clauses, and (2), to verify some of our constructions (see appendix). Finally, we combine several instances of the enforcer~$\mathcal{M}^{(i)}(u_1, \overline{u_2}, \overline{u_3})$ to obtain an unsatisfiable instance of \textsc{Monotone 3-Sat-$(2,2)$}. 

We start with the construction of the first gadget. Let $\mathcal{F}_2$ denote the set consisting of the following 2-clauses: 
\begin{multicols}{4}
\begin{enumerate}
\item $\{x_1, x_{2}\}$ 
\item $\{\overline{x_2}, \overline{x_3}\}$
\item $\{\overline{x_2}, \overline{x_4}\}$
\end{enumerate}
\end{multicols}
Further let  $\mathcal{F}_3$ denote the set consisting of the following 3-clauses:

\begin{multicols}{4}
\begin{enumerate}
\setcounter{enumi}{3}
\item $\{\overline{x_{3}}, \overline{x_{5}}, \overline{x_{6}}\}$
\item $\{\overline{x_{4}}, \overline{x_{5}}, \overline{x_{6}}\}$
\item $\{x_5, x_7, x_8\}$
\item $\{x_6, x_7, x_8\}$
\item $\{\overline{x_{7}}, \overline{z_1}, \overline{z_2}\}$
\item $\{\overline{x_{7}}, \overline{z_3}, \overline{z_4}\}$
\item $\{\overline{x_{8}}, \overline{z_1}, \overline{z_2}\}$
\item $\{\overline{x_{8}}, \overline{z_3}, \overline{z_4}\}$
\end{enumerate}
\end{multicols}

First, the 2-clauses in $\mathcal{F}_2$ are equivalent to the implications
\[
\overline{x_1} \Rightarrow x_2, \quad x_2 \Rightarrow \overline{x_3}, \quad x_2 \Rightarrow \overline{x_4}.
\]

Hence, if $\beta(x_1) = F$ then $\beta(x_2) = T$ and consequently $\beta(x_3) = \beta(x_4) = F$. Next, we introduce a set of clauses for which no satisfying truth assignment exists that sets $\beta(x_3) = F$ and $\beta(x_4) = F$. To this end, let $\mathcal{G}$ be the set consisting of the following 3-clauses: 

\begin{multicols}{4}
\begin{enumerate}
\setcounter{enumi}{11}
\item $\{x_3, y_1, y_2\}$
\item $\{x_3, y_3, y_4\}$
\item $\{x_4, y_5, y_6\}$
\item $\{x_4, y_7, y_8 \}$
\item $\{y_1, y_4, y_7\}$
\item $\{y_2, y_5, y_9\}$
\item $\{y_3, y_8, y_9\}$
\item $\{\overline{y_1}, \overline{y_5}, \overline{y_8}\}$
\item $\{\overline{y_1}, \overline{y_6}, \overline{y_9}\}$
\item $\{\overline{y_2}, \overline{y_3}, \overline{y_6}\}$
\item $\{\overline{y_2}, \overline{y_4}, \overline{y_8}\}$
\item $\{\overline{y_3}, \overline{y_5}, \overline{y_7}\}$
\item $\{\overline{y_4}, \overline{y_7}, \overline{y_9}\}$
\end{enumerate}
\end{multicols}

Note that for $\beta(x_3) = \beta(x_4) = F$, omitting the appearances of $x_3$ and $x_4$ in $\mathcal{G}$ has no effect on the satisfiability. We deferred the proof that the resulting instance is unsatisfiable to the appendix (see Lemma~\ref{lem:second_gadget}). Now, for at least one $x_i \in \{x_3, x_4\}$ we have $\beta(x_i) = T$ and we may assume that $\beta(x_1) = T$ and $\beta(x_2) = F$. Next, by clauses~4 and~5 we have $\beta(x_j) = F$ for at least one $x_j \in \{x_5, x_6\}$. Then, clauses~6 and~7 imply $\beta(x_k) = T$ for at least one $x_k \in \{x_7, x_8\}$. Hence, by clauses 8, 9, 10 and 11 we get two clauses~$\{F, \overline{z_1}, \overline{z_2}\}$ and~$\{F, \overline{z_3}, \overline{z_4}\}$ which is equivalent to $\{\overline{z_1}, \overline{z_2}\}$ and~$\{\overline{z_3}, \overline{z_4}\}$. Recalling that $\beta(x_1) = T$ and $\beta(x_2) = F$, the first three clauses in the following set~$\mathcal{H}$ of 3-clauses evaluate to $\{F, \overline{z_5}, \overline{z_6}\}$, $\{F, \overline{z_7}, \overline{z_8}\}$ and $\{F, z_7, z_{15}\}$, respectively.

\begin{multicols}{3}
\begin{enumerate}
\setcounter{enumi}{24}
\item $\{\overline{x_1}, \overline{z_5}, \overline{z_6}\}$
\item $\{\overline{x_1}, \overline{z_7}, \overline{z_8}\}$
\item $\{x_2, z_7, z_{15}\}$
\item $\{z_1, z_6, z_8\}$
\item $\{z_1, z_{11}, z_{12}\}$
\item $\{z_2, z_{6}, z_{8}\}$
\item $\{z_2, z_{11}, z_{12}\}$
\item $\{z_3, z_{5}, z_{9}\}$
\item $\{z_3, z_{13}, z_{14}\}$
\item $\{z_4, z_{5}, z_{14}\}$
\item $\{z_4, z_{9}, z_{10}\}$
\item $\{z_7, z_{10}, z_{13}\}$
\item $\{\overline{z_{5}}, \overline{z_{8}}, \overline{z_{15}}\}$
\item $\{\overline{z_6}, \overline{z_7}, \overline{z_9}\}$
\item $\{\overline{z_9}, \overline{z_{11}}, \overline{z_{13}}\}$
\item $\{\overline{z_{10}}, \overline{z_{11}}, \overline{z_{14}}\}$
\item $\{\overline{z_{10}}, \overline{z_{12}}, \overline{z_{14}}\}$
\item $\{\overline{z_{12}}, \overline{z_{13}}, \overline{z_{15}}\}$
\end{enumerate}
\end{multicols}

Now, the inferred 2-clauses 
\[
\{\overline{z_1}, \overline{z_2}\},\, \{\overline{z_3}, \overline{z_4}\},\, \{\overline{z_5}, \overline{z_6}\},\, \{\overline{z_7}, \overline{z_8}\} \text{ and } \{z_7, z_{15}\}
\]
in conjunction with the clauses $\mathcal{H} \setminus \{\{\overline{x_1}, \overline{z_5}, \overline{z_6}\}, \{\overline{x_1}, \overline{z_7}, \overline{z_8}\}, \{x_2, z_7, z_{15}\}\}$ are unsatisfiable (again, the proof is deferred to the appendix; see Lemma~\ref{lem:third_gadget}). 

Hence, the constructed set of 42 clauses 
\[
\mathcal{M} := \{\{x_1, x_{2}\}, \{\overline{x_2}, \overline{x_3}\}, \{\overline{x_2}, \overline{x_4}\}\} \cup \mathcal{F}_3 \cup \mathcal{G} \cup \mathcal{H}
\]
over the set of variables $V := X_1^8 \cup Y_1^9 \cup Z_1^{15}$ is unsatisfiable. We note that each literal appears at most twice in $\mathcal{M}$. The only variables that appear less than 4 times are $x_1, x_5, x_6, y_6$ and $z_{15}$ each of which appear once unnegated and twice negated. Consider the following enforcer
\[
\mathcal{M}^{(i)}(u_1, \overline{u_2}, \overline{u_3}) := \{\{x_1^i, x_{2}^i, u_1^i\}, \{\overline{x_2^i}, \overline{x_3^i}, \overline{u_2^i}\}, \{\overline{x_2^i}, \overline{x_4^i}, \overline{u_3^i}\}\} \cup \mathcal{F}_3^i \cup \mathcal{G}^i \cup \mathcal{H}^i,
\]
where $\mathcal{F}_3^i, \mathcal{G}^i, \mathcal{H}^i$ is obtained from $\mathcal{F}_3, \mathcal{G}, \mathcal{H}$ by replacing each variable, say $v$, with $v^i$ (e.g. $z_1$ is replaced with $z_1^i$). The enforcer $\mathcal{M}^{(i)}(u_1, \overline{u_2}, \overline{u_3})$ has two properties that we use to construct an unsatisfiable instance of \textsc{Monotone 3-Sat-$(2,2)$}. First, as alluded to in Section~\ref{sec:preliminaries}, we can deal with duplicates in a clause $\{u_1, \overline{u_2}, \overline{u_3}\}$, i.e., if $\overline{u_2} = \overline{u_3}$ and, second, we can transform a mixed clause into a monotone clause. Further, we obtain a second enforcer $\overline{\mathcal{M}}^{(i)}(\overline{u_1}, u_2, u_3)$ by negating every literal in $\mathcal{M}^{(i)}(u_1, \overline{u_2}, \overline{u_3})$. 

It is easy to verify that the following collection of clauses is unsatisfiable (we do not use set notation here since the clauses contain duplicates):
\begin{align*}
(\overline{a} \vee \overline{d} \vee \overline{f}) &\wedge (b \vee d \vee e) \wedge (e \vee \overline{b} \vee \overline{b}) \wedge (d \vee \overline{f} \vee \overline{c}) \\ 
&\wedge (a \vee \overline{c} \vee \overline{e}) \wedge (\overline{e} \vee c \vee c) \wedge (\overline{d} \vee a \vee b) \wedge (\overline{a} \vee f \vee f). 
\end{align*}

Now, we are in a position to construct an unsatisfiable instance~$\mathcal{U}$ of \textsc{Monotone 3-Sat-$(2,2)$}: 

\begin{align*}
\mathcal{U} := \{\{\overline{a}, \overline{d}, \overline{f}\}, \{b, d, e\}\} &\cup \mathcal{M}^{(1)}(e, \overline{b}, \overline{b}) \cup \mathcal{M}^{(2)}(d, \overline{f}, \overline{c}) \cup \mathcal{M}^{(3)}(a, \overline{c}, \overline{e}) \\
&\cup \overline{\mathcal{M}}^{(4)}(\overline{e}, c, c) \cup \overline{\mathcal{M}}^{(5)}(\overline{d}, a, b) \cup \overline{\mathcal{M}}^{(6)}(\overline{a}, f, f) \\
&\cup \bigcup_{i \in \{1,5,6\}} \{\{x_i^1, x_i^2, x_i^3\}, \{\overline{x_i^4}, \overline{x_i^5}, \overline{x_i^6}\}\}\\
&\cup \{\{y_6^1, y_6^2, y_6^3\}, \{\overline{y_6^4}, \overline{y_6^5}, \overline{y_6^6}\}, \{z_{15}^1, z_{15}^2, z_{15}^3\}, \{\overline{z_{15}^4}, \overline{z_{15}^5}, \overline{z_{15}^6}\}\}		
\end{align*}

\begin{prop}
There is an unsatisfiable instance of \textsc{Monotone 3-Sat-$(2,2)$} with 198 variables and 264 clauses. 
\end{prop}

Now, with the result from Darmann and Döcker~\cite[Thm.\,4]{darmann19} we get the following theorem as a consequence of the existence of an unsatisfiable instance of \textsc{Monotone 3-Sat-$(2,2)$}. 

\begin{thm}\label{thm:mon_3sat-(2,2)}
\textsc{Monotone 3-Sat-$(2,2)$} is NP-complete.
\end{thm}

Since \textsc{Monotone 3-Sat-$(k,k)$} is known to be NP-complete for each fixed $k \geq 3$~\cite[Cor.\,2]{darmann19}, we get the following corollary.  

\begin{cor}
\textsc{Monotone 3-Sat-$(k,k)$} is NP-complete for each fixed $k \geq 2$.
\end{cor}

\subsection*{A special case that is always satisfiable}

We briefly consider instances of \textsc{Monotone 3-Sat-$(k,k)$} with the property that for each clause $C = \{x, y, z\}$ the instance also contains $\overline{C} = \{\overline{x}, \overline{y}, \overline{z}\}$. Noting that this is \textsc{Monotone NAE 3-SAT} with exactly $k$ appearances of each variable, it follows that this problem is hard for $k = 4$ (see \cite[Cor.\,1]{darmann19}). 

\noindent {\bf Remark.} In the context of \textsc{NAE SAT} monotone means that negations are completely absent. This is no restriction since the two clauses $\{x, y, z\}$ and $\{\overline{x}, \overline{y}, \overline{z}\}$ impose exactly the same restrictions in this setting.

Porschen et al.~\cite[Thm.\,4]{porschen04} show that for $k = 3$ the corresponding \textsc{Monotone NAE 3-SAT} problem can be solved in linear time. In particular, they show that such an instance is nae-satisfiable if and only if the \emph{variable graph} has no component isomorphic to the complete graph $K_7$ on 7 vertices~\cite[Cor.\,4]{porschen04}. The variable graph (cf., e.g.,~\cite[p.\,2]{jain10} and~\cite[p.\,175]{porschen04}) of an instance of \textsc{NAE 3-SAT} (resp. \textsc{3-SAT}), contains a vertex for each variable and an edge between two vertices if the corresponding variables appear together in some clause of the instance. For example, the variable graph of the following instance is isomorphic to the $K_7$ and is, thus, not nae-satisfiable: 
\begin{align*}
\mathcal{U}_\text{NAE} = \{&\{x_1, x_2, x_7\}, \{x_1, x_3, x_6\}, \{x_1, x_4, x_5\}, \\ 
&\{x_2, x_3, x_4\}, \{x_2, x_5, x_6\}, \{x_3, x_5, x_7\}, \{x_4, x_6, x_7\}\}. 
\end{align*}

Let us now consider $k = 2$. We show that the property mentioned above leads to a trivial instance of \textsc{Monotone NAE 3-SAT} with exactly two appearances of each variable and, hence, \textsc{Monotone 3-Sat-$(2,2)$} is always satisfiable if clauses always appear in pairs $\{C, \overline{C}\}$. Jain~\cite[p.\,2]{jain10} observed that instances of \textsc{Monotone NAE 3-SAT} are in P if the variable graph is 4-colorable. Indeed, such instances are trivial since we can associate each truth value with exactly two colors such that a 4-coloring corresponds to a truth assignment that sets at least one variable of each clause false and at least one true (since each clause contains exactly three distinct variables, all clauses are satisfied). Pilz~\cite[Thm.\,12]{pilz19} used an approach based on this idea to show that every instance of \textsc{Planar SAT} in which each clause contains at least three negated or at least three unnegated appearances of distinct variables is satisfiable. He transformed the incidence graph of the formula into a certain subgraph of the variable graph, showed that this transformation preserves planarity, and then applied the Four Color Theorem~\cite{appel89} to obtain a 4-coloring.  
Hence, all we need to show is that the variable graph of an instance of \textsc{Monotone NAE 3-SAT} where each variable appears exactly twice is always 4-colorable. First, observe that a vertex corresponding to a variable $x$ in the variable graph of such an instance has degree 2 if and only if $x$ is contained in two clauses
\[
\{x, y, z\}, \{x, y, z\}
\] 
for some variables $y, z$ such that $x, y, z$ are pairwise distinct (otherwise $x$ has at least three neighbours). Such clauses can simply be removed as it is trivial to nae-satisfy them. Hence, we can assume that the variable graph has no cycles and, in particular, no cycles of odd length. Furthermore, it is easy to see that each instance has a number of variables that is divisible by 3 and, hence, each connected component in the variable graph contains a number of vertices that is a multiple of 3. Now, there is no component with 3 vertices since we already removed the clauses that would result in such a subgraph (the $K_3$ is a cycle of odd length). Noting that the degree of each vertex is bounded by 4, we conclude that no component with 6 or more vertices is a complete graph. Consequently, we can assume that the variable graph of an instance of \textsc{Monotone NAE 3-SAT} does not contain a component that is a complete graph or a cycle of odd length. Hence, the variable graph is 4-colorable by Brooks' Theorem~\cite{brooks41} and we get the following theorem.

\begin{thm}
All instances of \textsc{Monotone NAE 3-SAT}, where each variable appears exactly twice, are satisfiable.
\end{thm}

\begin{cor}
Let $\mathcal{I} = \bigcup_{j=1}^m \{C_j\}$ be an instance of \textsc{Monotone 3-Sat-$(2,2)$}. If the instance $\mathcal{I}$ has the property
\[
C_j \in \mathcal{I} \Rightarrow \overline{C_j} \in \mathcal{I},
\]
where $\overline{C_j}$ is obtained from $C_j$ by negating each literal, then $\mathcal{I}$ is satisfiable.  
\end{cor}

\section{More ways to obtain the main result}\label{sec:more_ways}

It is also possible to show NP-hardness of \textsc{Monotone 3-Sat-$(2,2)$} by reduction from \textsc{Monotone 3-Sat*-$(2,2)$}, for which NP-hardness was established by Darmann and Döcker~\cite[Thm.\,5]{darmann19}. To this end, let 
\[
\mathcal{N}^{(i)}(\overline{u_i}, \overline{u_i}) := \{\{x_1^i, x_{2}^i\}, \{\overline{x_2^i}, \overline{x_3^i}, \overline{u_i}\}, \{\overline{x_2^i}, \overline{x_4^i}, \overline{u_i}\}\} \cup \mathcal{F}_3^i \cup \mathcal{G}^i \cup \mathcal{H}^i,
\]
By construction, this set of clauses is not satisfied for any truth assignment $\beta$ that sets $\beta(u_1) = T$. Now, we can construct another enforcer which has exactly three positive 2-clauses:
\begin{align*}
\mathcal{S}(v_1, v_2, v_3) = &\{\{x_1^1, x_{2}^1, v_1\}, \{x_1^2, x_{2}^2, v_2\}, \{x_1^3, x_{2}^3, v_3\}\}\\
&\cup \mathcal{N}^{(1)}(\overline{u_1}, \overline{u_1})\setminus\{\{x_1^1, x_{2}^1\}\} \cup \mathcal{N}^{(2)}(\overline{u_2}, \overline{u_2})\setminus\{\{x_1^2, x_{2}^2\}\}\\
&\cup \mathcal{N}^{(3)}(\overline{u_3}, \overline{u_3})\setminus\{\{x_1^3, x_{2}^3\}\} \cup \{\{u_1, u_2, u_3\}\}\\
&\cup \bigcup_{i \in \{1,5,6\}} \{\{x_i^1, x_i^2, x_i^3\}\}\\
&\cup \{\{y_6^1, z_{15}^1, u_1\}, \{y_6^2, z_{15}^2, u_2\}, \{y_6^3, z_{15}^3, u_3\}\}
\end{align*}
Let $V_\mathcal{S}$ denote the set of variables that appear in $\mathcal{S}(v_1, v_2, v_3)$. Each variable $v \in V_\mathcal{S}\setminus\{v_1, v_2, v_3\}$ appears exactly twice unnegated and twice negated. For each instance of $\mathcal{S}(v_1, v_2, v_3)$, we create new variables $V_\mathcal{S}\setminus\{v_1, v_2, v_3\}$ (we omitted additional indices to improve readability). By negating each literal in $v \in V_\mathcal{S}\setminus\{v_1, v_2, v_3\}$ we obtain a second enforcer $\overline{\mathcal{S}}(\overline{v_1}, \overline{v_2}, \overline{v_3})$. By construction, the enforcer~$\mathcal{S}(v_1, v_2, v_3)$ has no satisfying truth assignment $\beta$ with $\beta(v_1) = \beta(v_2) = \beta(v_3) = F$. On the other hand, if $\beta(v_i) = T$ for at least one $v_i \in \{v_1, v_2, v_3\}$, we can assign truth values to the remaining variables of $\mathcal{S}(v_1, v_2, v_3)$ such that all clauses of the enforcer are satisfied (this is straightforward to verify with a SAT solver).

Given an instance $\mathcal{I}$ of \textsc{Monotone 3-Sat*-$(2,2)$}, we replace each positive (resp. negative) clause with a duplicate, say $(p \vee p \vee q)$ (resp. $(\overline{p} \vee \overline{p} \vee \overline{q})$), by an enforcer $\mathcal{S}(p, p, q)$ (resp. $\overline{\mathcal{S}}(\overline{p}, \overline{p}, \overline{q})$). The result is an instance of \textsc{Monotone 3-Sat-$(2,2)$} that is satisfiable if and only if $\mathcal{I}$ is satisfiable. 

Yet another approach is the following. We can also reduce from \textsc{3-Sat-$(2,2)$}, for which NP-hardness was established by Berman et al.~\cite[Thm.\,1]{berman03}, and use an extended version of the enforcers $\mathcal{M}^{(i)}(u_1, \overline{u_2}, \overline{u_3})$ and $\overline{\mathcal{M}}^{(i)}(\overline{u_1}, u_2, u_3)$ to transform mixed clauses that may be present in a given instance into monotone clauses. To this end, consider
\begin{align*}
\mathfrak{M}_{j} := &\mathcal{M}^{(3j)}(u_1, \overline{u_2}, \overline{u_3}) \cup \mathcal{M}^{(3j+1)}(u_4, \overline{u_5}, \overline{u_6}) \cup \mathcal{M}^{(3j+2)}(u_7, \overline{u_8}, \overline{u_9})\\
&\cup \{\{x_1^{3j}, x_5^{3j}, x_6^{3j}\},\{y_6^{3j}, z_{15}^{3j}, x_1^{3j+1}\}, \{x_5^{3j+1}, x_6^{3j+1}, y_6^{3j+1}\}\}\\
&\cup \{\{z_{15}^{3j+1}, x_1^{3j+2}, x_5^{3j+2}\}, \{x_6^{3j+2}, y_6^{3j+2}, z_{15}^{3j+2}\}\}
\end{align*}
Combining three instances of the enforcer $\mathcal{M}^{(i)}(u_1, \overline{u_2}, \overline{u_3})$ in this way has the advantage that each instance of $\mathfrak{M}_{j}$ introduces only variables that appear exactly twice unnegated and twice negated. A second enforcer $\overline{\mathfrak{M}_{j}}$ is again obtained by negating all literals. In order to be able to use these enforcers to replace all mixed clauses in a given instance of \textsc{3-Sat-$(2,2)$} we need the number of clauses with a positive (resp. negative) duplicate to be divisible by 3. This can be achieved by simply taking three copies of the original instance on pairwise disjoint sets of variables. With the help of a SAT solver it is easy to verify that $\mathfrak{M}_j$ has only satisfying truth assignments that set at least one literal in each of $\{u_1, \overline{u_2}, \overline{u_3}\}$, $\{u_4, \overline{u_5}, \overline{u_6}\}$ and $\{u_7, \overline{u_8}, \overline{u_9}\}$ true.   

\section{On a restricted variant of $\forall\exists$ \textsc{3-SAT}}\label{sec:quantified_results}

In this section, we consider the monotone variant of the following problem and show that it remains $\Pi_2^P$-complete in restricted settings. We assume the reader is familiar with basic concepts regarding the polynomial hierarchy and, in particular, with the complexity class $\Pi_2^P$. For an in-depth introduction to this theory, we refer to Stockmeyer~\cite{stockmeyer76} (see \cite{schaefer02} for a list containing many problems that are known to be $\Pi_2^P$-complete). We use the same notation defined in~\cite{{doecker19}}, e.g., for $i \leq i'$, let
\[
X_i^{i'} := \{x_i, x_{i+1}, \ldots, x_{i'}\},  
\]
and 
\[
Q X_i^{i'} := Qx_i Qx_{i+1} \cdots Qx_{i'}, \quad Q \in \{\forall, \exists\}.
\]

Let $s_1, s_2, t_1,t_2$ be four non-negative integers.     
\begin{center}
\noindent\fbox{\parbox{.95\textwidth}{
\noindent {\sc Balanced $\forall \exists$ 3-SAT-$(s_1,s_2, t_1, t_2)$}~\cite[p.\,6f]{doecker19}\\
\noindent{\bf Input.} A quantified Boolean formula
$$\forall X_1^p \exists X_{p+1}^n \bigcup_{j=1}^m \{C_j\}$$
over a set  $V=\{x_1,x_2,\ldots,x_n\}$ of variables such that (i) $n = 2p$, (ii) each $C_j$ is a 3-clause that contains three \emph{distinct} variables, and (iii), amongst the clauses, each universal variable appears unnegated exactly $s_1$ times and negated exactly $s_2$ times, and each existential variable appears unnegated exactly $t_1$ times and negated exactly $t_2$ times.
\\
\noindent{\bf Question.} For every truth assignment for $\{x_1, x_2, \ldots, x_p\}$, does there exist a truth assignment for $\{x_{p+1}, x_{p+2}, \ldots, x_n\}$ such that each clause of the formula is satisfied?
}}
\end{center}

Recently, Döcker et al.~\cite[Thm.\,3.1 and Thm.\,3.2]{doecker19} showed that {\sc Balanced $\forall \exists$ 3-SAT-$(2, 2, 2, 2)$} and {\sc Balanced $\forall \exists$ 3-SAT-$(1, 1, 2, 2)$} are both $\Pi_2^P$-complete. We use the gadgets $\mathfrak{M}_j$ and $\overline{\mathfrak{M}_j}$ to show that these results also hold for instances, where each clause is monotone (i.e., each clause consists of exactly three unnegated variables or exactly three negated variables, respectively). Since the transformation is virtually identical for both cases, we focus on the second result and mention the necessary adaption to obtain the first result. Consider an instance of {\sc Balanced $\forall \exists$ 3-SAT-$(1, 1, 2, 2)$}, i.e. a quantified Boolean formula
\[
\Phi = \forall X_1^p \exists X_{p+1}^n \varphi, 
\]
with $\varphi = \bigcup_{j=1}^m \{C_j\}$. Let $\varphi'$ and $\varphi''$ be the sets of clauses obtained from $\varphi$ by replacing $x_i$ with $y_i$ and $z_i$, respectively ($y_i$ and $z_i$ are distinct new variables). It is easy to see that the following quantified Boolean formula is a yes-instance if and only if $\Phi$ is a yes-instance. 
\[
\Phi' = \forall (X_1^p \cup Y_1^p \cup Z_1^p) \exists (X_{p+1}^n \cup Y_{p+1}^n \cup Z_{p+1}^n) (\varphi \cup \varphi' \cup \varphi'').  
\]
Now, the number of mixed clauses with two negative (resp. positive) literals is divisible by 3. Hence, we can replace such clauses in triples using $\mathfrak{M}_j$ and $\overline{\mathfrak{M}_j}$, respectively. For example, we replace first triple of mixed clauses, e.g., 
\[
\{x_i, \overline{x_j}, \overline{x_k}\}, \{y_i, \overline{y_j}, \overline{y_k}\}, \{z_i, \overline{z_j}, \overline{z_k}\}, 
\]
with the following collection of monotone clauses
\begin{align*}
\mathfrak{M}_{0} := &\mathcal{M}^{(0)}(x_i, \overline{x_j}, \overline{x_k}) \cup \mathcal{M}^{(1)}(y_i, \overline{y_j}, \overline{y_k}) \cup \mathcal{M}^{(2)}(z_i, \overline{z_j}, \overline{z_k})\\
&\cup \{\{x_1^{0}, x_5^{0}, x_6^{0}\},\{y_6^{0}, z_{15}^{0}, x_1^{1}\}, \{x_5^{1}, x_6^{1}, y_6^{1}\}\}\\
&\cup \{\{z_{15}^{1}, x_1^{2}, x_5^{2}\}, \{x_6^{2}, y_6^{2}, z_{15}^{2}\}\}.
\end{align*}

Note that we introduce $3 \cdot 32 = 96$ new existential variables with each instance of $\mathfrak{M}_j$ or $\overline{\mathfrak{M}_j}$. By construction, the resulting quantified Boolean formula $\Phi''$ is a yes-instance if and only if $\Phi'$ is a yes-instance. Since we introduced a number of existential variables that is divisible by 3, we can use multiple instances (each with new variables) of the following quantified enforcer introduced by Döcker et al.~\cite[p.\,9]{doecker19} 
\begin{align*}
Q^3 = &\{u, r, a\}, \{\overline{u}, \overline{b}, \overline{a}\}, \{v, q, b\}, {} \{\overline{v}, \overline{r}, \overline{a}\} , \{w, a, b\}, \{\overline{w}, \overline{q}, \overline{b}\},
\end{align*}
where $u, v, w, q, r$ are universal variables and $a, b$ are existential variables, to obtain a quantified Boolean formula with the same number of existential and universal variables. Since $Q^3$ is a yes-instance~\cite[Lem.\,3.2]{doecker19}, the resulting quantified Boolean formula is a yes-instance if and only if $\Phi''$ is a yes-instance. Noting that the transformation is polynomial, we get the following theorem. 

\begin{thm}
{\sc Balanced Monotone $\forall \exists$ 3-SAT-$(1, 1, 2, 2)$} is $\Pi_2^P$-complete.
\end{thm}

The only difference in the reduction from {\sc Balanced $\forall \exists$ 3-SAT-$(2, 2, 2, 2)$} to obtain the first result is the last step. Here, we are not able to use the existing quantified enforcer $Q^1$ given in~\cite[p.\,9]{doecker19}, since it introduces mixed clauses. For this reason, we adapt the quantified enforcer $Q^3$ as follows
\begin{align*}
Q^1_\text{mon} = &\{u, r, a\}, \{\overline{u}, \overline{b}, \overline{a}\}, \{v, q, b\}, \{\overline{v}, \overline{r}, \overline{a}\} , \{w, a, b\}, \{\overline{w}, \overline{q}, \overline{b}\}, {} \\   
&\{u, r, c\}, \{\overline{u}, \overline{d}, \overline{c}\}, \{v, q, d\}, \{\overline{v}, \overline{r}, \overline{c}\} , \{w, c, d\}, \{\overline{w}, \overline{q}, \overline{d}\},   
\end{align*}
where $u, v, w, q, r$ are universal variables and $a, b, c, d$ are existential variables. Intuitively, we use two instances of $Q^3$ on the same universal variables but with different existential variables. Consider an arbitrary truth assignment $\beta$ for the universal variables. Since $Q^3$ is a yes-instance we can find truth values $\beta(a)$ and $\beta(b)$ such that the top eight clauses in $Q^1_\text{mon}$ are satisfied. Hence, for $\beta(c) = \beta(a)$ and $\beta(d) = \beta(b)$ we can satisfy all clauses in $Q^1_\text{mon}$. In other words, $Q^1_\text{mon}$ is a yes-instance of $\forall \exists$ \textsc{3-SAT} that introduces 5 universal variables but only 4 existential variables (each of which appears exactly twice unnegated and exactly twice negated). Now, we can use multiple instances (each with new variables) of $Q^1_\text{mon}$ to obtain a formula with the same number of existential and universal variables. Thus, we get the following theorem. 

\begin{thm}
{\sc Balanced Monotone $\forall \exists$ 3-SAT-$(2, 2, 2, 2)$} is $\Pi_2^P$-complete.
\end{thm}

\appendix

\section{Proofs}

We used the PySAT Toolkit~\cite{ignatiev18} in the proofs of Lemmas~\ref{lem:second_gadget} and~\ref{lem:third_gadget} to obtain a DRUP proof~\cite{heule13} which is a certificate of unsatisfiablity. Here, we use the solver Lingeling~\cite{biere16} included in the PySAT Toolkit since it is one of the solvers that provide the option to return such a certificate of unsatisfiability. 

\begin{lem}\label{lem:second_gadget} The following set of clauses over variables $Y_1^9$ is unsatisfiable. 
\begin{multicols}{4}
\begin{enumerate}
\item $\{y_1, y_2\}$
\item $\{y_3, y_4\}$
\item $\{y_5, y_6\}$
\item $\{y_7, y_8\}$
\item $\{y_1, y_4, y_7\}$
\item $\{y_2, y_5, y_9\}$
\item $\{y_3, y_8, y_9\}$
\item $\{\overline{y_1}, \overline{y_5}, \overline{y_8}\}$
\item $\{\overline{y_1}, \overline{y_6}, \overline{y_9}\}$
\item $\{\overline{y_2}, \overline{y_3}, \overline{y_6}\}$
\item $\{\overline{y_2}, \overline{y_4}, \overline{y_8}\}$
\item $\{\overline{y_3}, \overline{y_5}, \overline{y_7}\}$
\item $\{\overline{y_4}, \overline{y_7}, \overline{y_9}\}$
\end{enumerate}
\end{multicols}
\end{lem}

\begin{proof}
We can use the following Python code to obtain a DRUP proof. 
\begin{lstlisting}[frame=single]
from pysat.solvers import Lingeling
cnf = [[1, 2], [3, 4], [5, 6], [7, 8], [1, 4, 7], [2, 5, 9], 
       [3, 8, 9], [-1, -5, -8], [-1, -6, -9], [-2, -3, -6], 
       [-2, -4, -8], [-3, -5, -7], [-4, -7, -9]]
solver = Lingeling(bootstrap_with=cnf, with_proof=True)
print solver.solve()
print solver.get_proof()
solver.delete()
\end{lstlisting}
Output of the program:
\begin{lstlisting}[language=bash]
False
['-8 -7 -5 0', '-8 9 5 0', '-5 -8 0', 'd -1 -5 -8 0', '-8 9 0', 'd 5 -8 9 0', '9 0', '-4 -2 0', 'd -8 -4 -2 0', '-5 -3 0', 'd -7 -5 -3 0', '-3 -2 0', 'd -6 -3 -2 0', '-2 0', '1 0', '-6 0', '5 0', '-8 0', '-3 0', '7 0', '4 0', '0']
\end{lstlisting}
\end{proof}

\begin{lem}\label{lem:third_gadget} The following set of clauses over variables $Z_1^{15}$ is unsatisfiable.    

\begin{multicols}{3}
\begin{enumerate}
\item $\{\overline{z_1}, \overline{z_2}\}$
\item $\{\overline{z_3}, \overline{z_4}\}$
\item $\{\overline{z_5}, \overline{z_6}\}$
\item $\{\overline{z_7}, \overline{z_8}\}$
\item $\{z_7, z_{15}\}$
\item $\{z_1, z_6, z_8\}$
\item $\{z_1, z_{11}, z_{12}\}$
\item $\{z_2, z_{6}, z_{8}\}$
\item $\{z_2, z_{11}, z_{12}\}$
\item $\{z_3, z_{5}, z_{9}\}$
\item $\{z_3, z_{13}, z_{14}\}$
\item $\{z_4, z_{5}, z_{14}\}$
\item $\{z_4, z_{9}, z_{10}\}$
\item $\{z_7, z_{10}, z_{13}\}$
\item $\{\overline{z_{5}}, \overline{z_{8}}, \overline{z_{15}}\}$
\item $\{\overline{z_6}, \overline{z_7}, \overline{z_9}\}$
\item $\{\overline{z_9}, \overline{z_{11}}, \overline{z_{13}}\}$
\item $\{\overline{z_{10}}, \overline{z_{11}}, \overline{z_{14}}\}$
\item $\{\overline{z_{10}}, \overline{z_{12}}, \overline{z_{14}}\}$
\item $\{\overline{z_{12}}, \overline{z_{13}}, \overline{z_{15}}\}$

\end{enumerate}
\end{multicols}
\end{lem}

\begin{proof}
We can use the following Python code to obtain a DRUP proof. 
\begin{lstlisting}[frame=single]
from pysat.solvers import Lingeling
cnf = [[-1, -2], [-3, -4], [-5, -6], [-7, -8], [7, 15], 
       [1, 6, 8], [1, 11, 12], [2, 6, 8], [2, 11, 12], 
       [3, 5, 9], [3, 13, 14], [4, 5, 14], [4, 9, 10], 
       [7, 10, 13], [-5, -8, -15], [-6, -7, -9], 
       [-9, -11, -13], [-10, -11, -14], [-10, -12, -14], 
       [-12, -13, -15]]
solver = Lingeling(bootstrap_with=cnf, with_proof=True)
print solver.solve()
print solver.get_proof()
solver.delete()
\end{lstlisting}
Output of the program:
\begin{lstlisting}[language=bash]
False
['6 8 0', 'd 1 6 8 0', '14 10 13 0', '11 12 0', 'd 1 11 12 0', '-8 14 -13 0', '14 -13 0', 'd -8 14 -13 0', '14 10 0', 'd 13 14 10 0', '-9 10 7 0', '-9 10 0', 'd 7 -9 10 0', '-5 0', '10 9 0', 'd 4 10 9 0', '-13 -15 0', 'd -12 -13 -15 0', '-14 -10 0', 'd -11 -14 -10 0', '14 13 0', 'd 3 14 13 0', '13 7 0', 'd 10 13 7 0', '7 0', '-8 0', '6 0', '-9 0', '3 0', '10 0', '-4 0', '-14 0', '0']
\end{lstlisting}
\end{proof}

\section{Enforcer $\mathcal{M}^{(i)}(u_1, \overline{u_2}, \overline{u_3})$}

The set of clauses $\mathcal{M}$ constructed in Section~\ref{sec:construction_unsatisfiable_instance} is the basis for the enforcer $\mathcal{M}^{(i)}(u_1, \overline{u_2}, \overline{u_3})$ and thus, for several results presented in this article. To facilitate verification of our results, we provide the set $\mathcal{M}$ as a Python list:\\

\begin{scriptsize}
[[1, 2], [-2, -3], [-2, -4], [-3, -5, -6], [-4, -5, -6], [5, 7, 8], [6, 7, 8], [-7, -18, -19], [-7, -20, -21], [-8, -18, -19], [-8, -20, -21], [3, 9, 10], [3, 11, 12], [4, 13, 14], [4, 15, 16], [9, 12, 15], [10, 13, 17], [11, 16, 17], [-9, -13, -16], [-9, -14, -17], [-10, -11, -14], [-10, -12, -16], [-11, -13, -15], [-12, -15, -17], [2, 24, 32], [18, 23, 25], [18, 28, 29], [19, 23, 25], [19, 28, 29], [20, 22, 26], [20, 30, 31], [21, 22, 31], [21, 26, 27], [24, 27, 30], [-1, -22, -23], [-1, -24, -25], [-22, -25, -32], [-23, -24, -26], [-26, -28, -30], [-27, -28, -31], [-27, -29, -31], [-29, -30, -32]]
\end{scriptsize}
   
\section{Unsatisfiable instance of Monotone 3-Sat-$(2,2)$}

The unsatisfiable instance constructed in Section~\ref{sec:construction_unsatisfiable_instance} as a Python list:\\

\begin{scriptsize}
[[-193, -196, -198], [194, 196, 197], [1, 2, 197], [-2, -3, -194], [-2, -4, -194], [-3, -5, -6], [-4, -5, -6], [5, 7, 8], [6, 7, 8], [-7, -18, -19], [-7, -20, -21], [-8, -18, -19], [-8, -20, -21], [3, 9, 10], [3, 11, 12], [4, 13, 14], [4, 15, 16], [9, 12, 15], [10, 13, 17], [11, 16, 17], [-9, -13, -16], [-9, -14, -17], [-10, -11, -14], [-10, -12, -16], [-11, -13, -15], [-12, -15, -17], [2, 24, 32], [18, 23, 25], [18, 28, 29], [19, 23, 25], [19, 28, 29], [20, 22, 26], [20, 30, 31], [21, 22, 31], [21, 26, 27], [24, 27, 30], [-1, -22, -23], [-1, -24, -25], [-22, -25, -32], [-23, -24, -26], [-26, -28, -30], [-27, -28, -31], [-27, -29, -31], [-29, -30, -32], [33, 34, 196], [-34, -35, -195], [-34, -36, -198], [-35, -37, -38], [-36, -37, -38], [37, 39, 40], [38, 39, 40], [-39, -50, -51], [-39, -52, -53], [-40, -50, -51], [-40, -52, -53], [35, 41, 42], [35, 43, 44], [36, 45, 46], [36, 47, 48], [41, 44, 47], [42, 45, 49], [43, 48, 49], [-41, -45, -48], [-41, -46, -49], [-42, -43, -46], [-42, -44, -48], [-43, -45, -47], [-44, -47, -49], [34, 56, 64], [50, 55, 57], [50, 60, 61], [51, 55, 57], [51, 60, 61], [52, 54, 58], [52, 62, 63], [53, 54, 63], [53, 58, 59], [56, 59, 62], [-33, -54, -55], [-33, -56, -57], [-54, -57, -64], [-55, -56, -58], [-58, -60, -62], [-59, -60, -63], [-59, -61, -63], [-61, -62, -64], [65, 66, 193], [-66, -67, -195], [-66, -68, -197], [-67, -69, -70], [-68, -69, -70], [69, 71, 72], [70, 71, 72], [-71, -82, -83], [-71, -84, -85], [-72, -82, -83], [-72, -84, -85], [67, 73, 74], [67, 75, 76], [68, 77, 78], [68, 79, 80], [73, 76, 79], [74, 77, 81], [75, 80, 81], [-73, -77, -80], [-73, -78, -81], [-74, -75, -78], [-74, -76, -80], [-75, -77, -79], [-76, -79, -81], [66, 88, 96], [82, 87, 89], [82, 92, 93], [83, 87, 89], [83, 92, 93], [84, 86, 90], [84, 94, 95], [85, 86, 95], [85, 90, 91], [88, 91, 94], [-65, -86, -87], [-65, -88, -89], [-86, -89, -96], [-87, -88, -90], [-90, -92, -94], [-91, -92, -95], [-91, -93, -95], [-93, -94, -96], [-97, -98, -197], [98, 99, 195], [98, 100, 195], [99, 101, 102], [100, 101, 102], [-101, -103, -104], [-102, -103, -104], [103, 114, 115], [103, 116, 117], [104, 114, 115], [104, 116, 117], [-99, -105, -106], [-99, -107, -108], [-100, -109, -110], [-100, -111, -112], [-105, -108, -111], [-106, -109, -113], [-107, -112, -113], [105, 109, 112], [105, 110, 113], [106, 107, 110], [106, 108, 112], [107, 109, 111], [108, 111, 113], [-98, -120, -128], [-114, -119, -121], [-114, -124, -125], [-115, -119, -121], [-115, -124, -125], [-116, -118, -122], [-116, -126, -127], [-117, -118, -127], [-117, -122, -123], [-120, -123, -126], [97, 118, 119], [97, 120, 121], [118, 121, 128], [119, 120, 122], [122, 124, 126], [123, 124, 127], [123, 125, 127], [125, 126, 128], [-129, -130, -196], [130, 131, 193], [130, 132, 194], [131, 133, 134], [132, 133, 134], [-133, -135, -136], [-134, -135, -136], [135, 146, 147], [135, 148, 149], [136, 146, 147], [136, 148, 149], [-131, -137, -138], [-131, -139, -140], [-132, -141, -142], [-132, -143, -144], [-137, -140, -143], [-138, -141, -145], [-139, -144, -145], [137, 141, 144], [137, 142, 145], [138, 139, 142], [138, 140, 144], [139, 141, 143], [140, 143, 145], [-130, -152, -160], [-146, -151, -153], [-146, -156, -157], [-147, -151, -153], [-147, -156, -157], [-148, -150, -154], [-148, -158, -159], [-149, -150, -159], [-149, -154, -155], [-152, -155, -158], [129, 150, 151], [129, 152, 153], [150, 153, 160], [151, 152, 154], [154, 156, 158], [155, 156, 159], [155, 157, 159], [157, 158, 160], [-161, -162, -193], [162, 163, 198], [162, 164, 198], [163, 165, 166], [164, 165, 166], [-165, -167, -168], [-166, -167, -168], [167, 178, 179], [167, 180, 181], [168, 178, 179], [168, 180, 181], [-163, -169, -170], [-163, -171, -172], [-164, -173, -174], [-164, -175, -176], [-169, -172, -175], [-170, -173, -177], [-171, -176, -177], [169, 173, 176], [169, 174, 177], [170, 171, 174], [170, 172, 176], [171, 173, 175], [172, 175, 177], [-162, -184, -192], [-178, -183, -185], [-178, -188, -189], [-179, -183, -185], [-179, -188, -189], [-180, -182, -186], [-180, -190, -191], [-181, -182, -191], [-181, -186, -187], [-184, -187, -190], [161, 182, 183], [161, 184, 185], [182, 185, 192], [183, 184, 186], [186, 188, 190], [187, 188, 191], [187, 189, 191], [189, 190, 192], [1, 33, 65], [-97, -129, -161], [5, 37, 69], [-101, -133, -165], [6, 38, 70], [-102, -134, -166], [14, 46, 78], [-110, -142, -174], [32, 64, 96], [-128, -160, -192]]
\end{scriptsize}


\begin{thebibliography}{00}

\bibitem{appel89}
K. Appel, and W. Haken (1989). Every Planar Map is Four Colorable. {\it In} Contemporary Mathematics, vol. 98, American Mathematical Soc.

\bibitem{berman03}
P. Berman, M. Karpinski, and A. D. Scott (2003). Approximation hardness of short symmetric instances of MAX-3SAT. {\it Electronic Colloquium on Computational Complexity}, Report No. 49.

\bibitem{biere16}
A. Biere (2016). Splatz, Lingeling, Plingeling, Treengeling, YalSAT entering the SAT competition 2016. {\it In} Proceedings of SAT Competition 2016, pp. 44--45.

\bibitem{brooks41}
R. L. Brooks (1941). On colouring the nodes of a network. {\it Mathematical Proceedings of the Cambridge Philosophical Society}, 37(2).

\bibitem{cook71}
S. A. Cook (1971). The complexity of theorem-proving procedures. {\it In} Proceedings of the Third Annual ACM Symposium on Theory of Computing, pp. 151--158. 

\bibitem{darmann19}
A. Darmann and J. D\"ocker (2019). On simplified NP-complete variants of Not-All-Equal 3-SAT and 3-SAT. arXiv preprint arXiv:1908.04198.

\bibitem{darmann18}
A. Darmann, J. D\"ocker, B. Dorn (2018). The Monotone Satisfiability Problem with Bounded Variable Appearances. {\it International Journal of Foundations of Computer Science}, 29(6):979--993.

\bibitem{doecker19}
J. D\"ocker, B. Dorn, S. Linz, C. Semple (2019). Placing quantified variants of 3-SAT and Not-All-Equal 3-SAT in the polynomial hierarchy. arXiv preprint arXiv:1908.05361.

\bibitem{garey79}
M. R. Garey and D. S. Johnson (1979). Computers and Intractability: A Guide to the Theory of NP-Completeness, W. H. Freeman and Company.

\bibitem{heule13}
M. J. H. Heule, W. A. Hunt, N. Wetzler (2013). Trimming while checking clausal proofs. {\it In} 2013 Formal Methods in Computer-Aided Design. IEEE (2013), pp. 181--188.

\bibitem{ignatiev18} A. Ignatiev, A. Morgado, J. Marques-Silva (2018). PySAT: A Python Toolkit for Prototyping with SAT Oracles. {\it In} O. Beyersdorff, C. Wintersteiger (eds) Theory and Applications of Satisfiability Testing – SAT 2018. SAT 2018. Lecture Notes in Computer Science, vol. 10929, Springer, pp. 428--437.

\bibitem{jain10} P. Jain (2010). On a variant of Monotone NAE-3SAT and the Triangle-Free Cut problem. arXiv preprint arXiv:1003.3704.

\bibitem{pilz19} A. Pilz. Planar 3-SAT with a clause/variable cycle. {\it Discrete Mathematics \& Theoretical Computer Science}, 21(3).

\bibitem{porschen04}
S. Porschen, B. Randerath, E. Speckenmeyer (2004). Linear time algorithms for some not-all-equal satisfiability problems. {\it In} E. Giunchiglia, A. Tacchella (eds) Theory and Applications of Satisfiability Testing. SAT 2003. Lecture Notes in Computer Science, vol. 2919, Springer, pp. 172--187.

\bibitem{schaefer02}
M. Schaefer and C. Umans (2002). Completeness in the polynomial-time hierarchy: A compendium. {\it SIGACT News}, 33:32--49.

\bibitem{stockmeyer76}
L. J. Stockmeyer (1976).  The polynomial-time hierarchy. {\it Theoretical Computer Science}, 3:1--22.

\bibitem{tovey84}
C. A. Tovey (1984). A simplified NP-complete satisfiability problem. {\it Discrete Applied Mathematics}, 8:85--89. 

\end{thebibliography}
\end{document}